\documentclass[entropy,article,submit,moreauthors,pdftex]{Definitions/mdpi}

\firstpage{1}
\pubvolume{xx}
\issuenum{x}
\articlenumber{x}
\pubyear{2022}
\copyrightyear{2022}
\history{Received: date; Accepted: date; Published: date}

\usepackage{color,xspace}
\usepackage{Definitions/macros_gpi}
\usepackage{graphicx}

\graphicspath{{Figs/}}

\setitemize{parsep=6pt,itemsep=0pt,leftmargin=*,labelsep=5.5mm}
\setenumerate{parsep=6pt,itemsep=0pt,leftmargin=*,labelsep=5.5mm} 
\setlist[description]{itemsep=0mm}   
\usepackage{environ}
\NewEnviron{myequation}{%
\begin{equation} 
\scalebox{0.95}{$\BODY$}
\end{equation}
}

\Title{Local Max-Entropy and Free Energy Principles solved
by Belief Propagation$^\dagger$}

\newcommand{\orcidauthorA}{0000-0003-1347-3790} 

\Author{Olivier Peltre$^{1}$ orcid number:\orcidauthorA} 

\AuthorNames{Olivier Peltre}

\address[1]{%
$^1$ \quad Université d'Artois, Faculté Jean Perrin (LML), Rue Jean Souvraz 62300 Lens
}

\firstnote{Submitted to International Workshop on Bayesian Inference and Maximum Entropy Methods in Science and Engineering, IHP, Paris, July 18-22, 2022.} 

\newcommand{\R}{\mathbb{R}}
\newcommand{\Z}{\mathbb{Z}}

\newcommand{\F}{F}
\newcommand{\varF}{\mathcal{F}}
\newcommand{\U}{\mathcal{U}}

\renewcommand{\Prob}{\mathrm{Prob}}
\newcommand{\E}{\mathbb{E}}

\newcommand{\Alg}{\mathbf{Alg}}
\newcommand{\Vect}{\mathbf{Vect}}
\newcommand{\Set}{\mathbf{Set}}

\newcommand{\Pow}{\mathcal{P}}
\newcommand{\incl}{\subseteq}
\newcommand{\cont}{\supseteq}
\newcommand{\eqvl}{\Leftrightarrow}

\DeclareMathOperator{\Ker}{Ker}
\DeclareMathOperator{\Img}{Im}

\DeclareMathOperator{\e}{\mathrm{e}}

\renewcommand{\aa}{\mathrm{a}}
\renewcommand{\bb}{\mathrm{b}}
\newcommand{\cc}{\mathrm{c}}
\newcommand{\Om}{\Omega}
\newcommand{\ph}{\varphi}

\newcommand{\Fix}{\mathrm{Fix}}
\newcommand{\Fixt}{\mathrm{Fix}_+}
\abstract{
A statistical system is classically defined on a set of microstates 
$E$ by a global energy function $H : E \to \R$, 
yielding Gibbs probability measures (softmins) $\rho^\beta(H)$
for every inverse temperature $\beta = T^{-1}$.
Gibbs states are simultaneously characterised by free energy principles 
and the max-entropy principle, 
with dual constraints on inverse temperature $\beta$
and mean energy $\U(\beta) = \E_{\rho^\beta}[H]$ respectively.
The Legendre transform relates these diverse variational principles \cite{Jaynes-57, Marle-16}
which are unfortunately not tractable in high dimension.
\\
We here consider a finite free sheaf $E : \Pow(\Om)^{op} \to \Set$ of microstates, 
defined by $E_\aa = \prod_{i \in \aa} E_i$ for every region $\aa \incl \Om$,
whose size grows exponentially with $\Om$.
The global energy function $H : E \to \R$ is assumed given as a sum 
$H(x) = \sum_{\aa \incl \Om} h_\aa(x_{|\aa})$
of local short-range interactions 
$h_\aa : E_\aa \to \R$, indexed by bounded subregions.
The cluster variational method (CVM) \cite{Kikuchi-51,Morita-57,Pelizzola-05} can then be used to estimate
the global free energy $F(\beta) = - \frac 1 {\beta} \ln \sum \e^{- \beta H}$, 
by searching
for critical points of a localised approximation of a 
variational free energy. 
These critical points can be found by generalised belief propagation
(GBP) algorithms, according to a correspondence theorem 
initially conjectured by Yedidia et al \cite{Yedidia-2001, Yedidia-2005}. 
The proof of this theorem \cite{gsi19, phd} involved 
combinatorial and topological structures acting on 
dual complexes of local observables and local measures, 
on which continuous-time diffusion equations gave new message-passing
algorithms regularising GBP \cite{phd, gsi21}.
\\
We now show that stationary states of GBP algorithms 
solve a collection of equivalent variational principles 
on local Bethe-Kikuchi functionals, 
which approximate the free energy $F(\beta)$, the 
Shannon entropy $S(\U)$, 
and the variational free energy $\varF(\beta) = \U - \beta^{-1} S(\U)$ respectively.
This local form of Legendre duality yields a possibly degenerate relationship 
between the constraints $\U$ and $\beta$, exhibiting singularities on loopy hypergraphs
where GBP converges to multiple equilibria.  
\nocite{Bethe-35, Kikuchi-51, Morita-57,Mezard-Montanari}
}


\begin{document}

\section{Local statistical systems: topology and combinatorics}

Consider a finite free sheaf $E : \Pow(\Om)^{op} \to \Set$ of microstates, 
defined for every region $\aa \incl \Om$ by 
$E_\aa = \prod_{i \in \aa} E_i$. For every $\bb \incl \aa$ 
the natural restriction $\pi^{\aa\to\bb} : E_\aa \to E_\bb$ forgets the state 
of variables outside $\bb$, we write $x_{\aa | \bb} = \pi^{\aa\to\bb}(x_\aa)$.
We recall necessary definitions and theorems from \cite{phd} in this section.

{\it Local observables} form a cosheaf\footnote{
    In the sense that $\R^{E_{\aa \sqcup \bb}} = \R^{E_\aa} \otimes \R^{E_\bb}$ 
    for every disjoint pair $\aa, \bb \incl \Om$. The tensor product $\otimes$ is
    a coproduct in the category of commutative algebras $\Alg_c$, and 
    $\R^{E_\aa} = \otimes_{i \in \aa} \R^{E_i}$ can be identified
    with a colimit in $\Alg_{c}$.  
} 
$\R^E : \Pow(\Om) \to \Alg_{c}$ of commutative algebras, 
whose arrows consist of natural inclusions $\R^{E_\bb} \incl \R^{E_\aa}$ 
when identifying each local algebra with a low dimensional subspace of $\R^E$.
Let us write $\tilde{h}^\aa_\bb(x_\aa) = h_\bb(x_{\aa|\bb})$ when
the extension should be made explicit.
Given a hypergraph $K \incl \Pow(\Om)$,  one can define a chain complex
of local observables \cite[Chapter 3]{phd} \cite{SGA-4-V} indexed by the nerve of $K$:
\begin{equation} 
    C_0(K, \R^E) \overset{\delta}{\longleftarrow} 
    C_1(K, \R^E) \overset{\delta}{\longleftarrow} \;\dots\; \overset{\delta}{\longleftarrow}
    C_n(K, \R^E)
\end{equation}
An element of $C_0(K, \R^E)$ is a collection $[h_\aa : E_\aa \to \R\,|\, \aa \in K]$ and 
an element of $C_1(K, \R^E)$ is a collection 
$[\ph_{\aa \to \bb} : E_\bb \to \R \,|\, \aa,\bb \in K,\, \aa \supset \bb ]$:
they are statistical analogs of energy densities and heat fluxes respectively.
The {\it boundary operator} $\delta : C_1(K, \R^E) \to C_0(K, \R^E)$ acts like a divergence: 
\begin{equation} 
    \delta\ph_\bb(x_\bb) = \sum_{\aa \supset \bb} \ph_{\aa \to \bb}(x_\bb)
    - \sum_{\bb \supset \cc} \ph_{\bb \to \cc}(x_{\bb|\cc})
\end{equation}
We showed in \cite{phd} that stationary points of GBP algorithms can be found by 
diffusion equations of the form $\frac {dh}{dt} = \delta \Phi(h)$, 
where $\Phi : C_0(K, \R^E) \to C_1(K, \R^E)$ 
is a flux functional satisfying some consistency constraints
\cite[thm 5.13]{phd}, while homology classes $[h] = h + \delta C_1(K, \R^E)$
explored by diffusion can be characterised by a global energy
conservation constraint $\sum_{\aa \in K} \tilde{h}^\Om_\aa = H_\Om$
\cite[thm 2.13]{phd}.

{\it Local densities} form a dual functor 
$\R^{E*}: \Pow(\Om)^{op} \to \Vect$ of vector spaces, whose arrows 
consist of partial integrations 
$\pi_*^{\aa \to \bb} : \R^{E_\aa *} \to \R^{E_\bb *}$,
also called marginal projections. In the
cochain complex of local densities $C_\bullet(K, \R^E)^*$,
the {\it differential operator} $d : C_0(K, \R^E)^* \to C_1(K, \R^E)^*$
acts by:
\begin{equation}
    d p_{\aa \to \bb}(x_\bb) = p_\bb(x_\bb) - \sum_{x_{\aa|\bb} =\, x_\bb} p_\aa(x_\aa)
\end{equation}
{\it Consistent densities} satisfying $dp = 0$ are hence characterised by 
cohomology classes $[p] \in \Ker(d)$ in $C_0(K, \R^E)^*$.
They can be extendend to a 
global measure $p_\Om$ such that 
$p_\aa = \pi_*^{\Om \to \aa}(p_\Om)$ for all $\aa \in K$, 
although a positive $p_\Om$ may not always exist for all positive and consistent 
$(p_\aa)_{\aa \in K}$. 

The {\it Gibbs correspondence}
relating potentials $h \in C_0(K, \R^E)$ 
and beliefs $p \in C_0(K, \R^E)^*$ is essential to 
the dynamic of GBP. Time evolution 
preserves the total energy $[h]$ while trying 
to enforce consistency constraints on $p$, given 
for some choice of inverse temperature $\beta$ by:
\begin{equation}
    p_\aa(x_\aa) = \frac 1 {Z_\aa} 
    \e^{- \beta H_\aa(x_\aa)}
    \quad\mathrm{with}\quad
    H_\aa(x_\aa) = \sum_{\bb \incl \aa} h_\bb(x_{\aa|\bb})
\end{equation}
We shall write $p = \rho^\beta(H)$ with $H = \zeta(h)$. 
Then $\rho^\beta$ is the non-linear 
Gibbs state map describing statistical equilibrium in physics, 
and the {\it zeta transform} $\zeta$, 
acting by local summation on subregions, is a combinatorial automorphism of 
$C_0(K, \R^E)$.
{\it Möbius inversion} 
formulas give an explicit characterization of the inverse automorphism $\mu = \zeta^{-1}$:
\begin{equation} \label{moebius}
    h_\aa(x_\aa) = \sum_{\bb \incl \aa} 
    \mu_{\aa \to \bb} \: H_{\bb}(x_{\aa|\bb}) 
\end{equation}
where the coefficients $\mu_{\aa\to\bb} \in \Z$ solve\footnote{
    The matrix $\mu$ can be computed efficiently by iterating powers 
    of the nilpotent operator $(\zeta - 1)$. When $K$ describes a graph
    and has only regions of size 2, 1 and 0, one for instance has $(\zeta - 1)^3 = 0$.
}
inclusion-exclusion principles on 
the partial order $(K, \incl)$. 
{\it Bethe-Kikuchi coefficients} are defined as 
$c_\bb = \sum_{\aa \cont \bb} \mu_{\aa \to \bb}$, 
or equivalently by the inclusion-exclusion principle 
$\sum_{\aa \cont \bb} c_\bb = 1$ for all $\bb \in K$.
A consequence of 
\ref{moebius} is that for any $x \in E$:
\begin{equation} \label{U_Bethe}
    \sum_{\aa \in K} h_\aa(x_{|\aa}) = 
    \sum_{\bb \in K} c_\bb \: H_\bb(x_{|\bb})
\end{equation} 
Combinatorial theory is at the heart of Bethe approximations and hence of the CVM.
It is related to a Dirichlet convolution structure on the partial order $(K, \incl)$ 
as well described in the work of Rota \cite{Rota-64}, who saw in Möbius inversion 
formulas a discrete analogy with the fundamental theorem of calculus.  
The reciprocal pair of automorphisms $(\zeta, \mu)$ can indeed be 
extended to the whole complex $C_\bullet(K, \R^E)$, with
$\zeta$ and $\delta$ satisfy commutation relations that resemble Stokes formulas
\cite[chap 3]{phd}.
\section{Cluster variational principles solved by diffusion}

In the following, assume $K \incl \Pow(\Om)$ is fixed and closed by intersection.
Let us introduce notations $C_\bullet = C_\bullet(K, \R^E)$ 
for the complex of observables, 
$\Delta_0 \incl C^*_0$ for the convex subspace of positive beliefs, 
constrained by $p_\aa \in \Prob(E_\aa)$ and $p_\aa > 0$ for all $\aa \in K$, 
and $\Gamma_0 = \Delta_0 \cap \Ker(d)$ for the subconvex of consistent beliefs, 
further constrained by $dp = 0$. 

In this section, we define variational principles on $\Gamma_0$ 
and characterise the dynamical equations one can be use to solve them. 
Because $\rho^\beta$ is not linear, 
the manifold $\Fix^\beta = (\rho^\beta \circ \zeta)^{-1}(\Gamma_0)$ 
defined below is not a linear subspace of $C_0$,
despite $\Gamma_0 \incl \Delta_0$ being a convex polytope of $C_0^*$.

\begin{Definition} 
We call {\em consistent manifold} the non-linear subspace 
$\Fix^\beta= \{ v \in C_0 \,|\, 
\rho^{\beta}(\zeta v) \in \Gamma_0 \}$ 
for any choice of inverse temperature $\beta > 0$, and let $\Fix = \Fix^1$.
\end{Definition}

\begin{Definition} 
We say that a flux functional $\Phi : C_0 \to C_1$ is 
{\em projectively faithful} at $\beta > 0$ iff:
\begin{equation} 
\delta \Phi(v) = 0 \Leftrightarrow v \in \Fix^\beta
\end{equation}
We call $\Phi$ {\em projectively consistent} if
it satisfies the weaker condition $\delta\Phi_{|\Fix^\beta} = 0$.
\end{Definition}

When $\Phi$ is projectively faithful at $\beta = 1$, stationary potentials of 
the diffusion $\frac {dv}{dt} = \delta \Phi(v)$ in $C_0$ are projected onto 
stationary beliefs $p = \rho(\zeta v) \in \Gamma_0$ of GBP 
\cite[thm 5.15]{phd}. 
The purpose of GBP is hence to find points at the intersection 
of energy conservation and consistency constraint surfaces, 
i.e. given initial potentials $h \in C_0$, GBP looks for $v \in [h] \cap \Fix$,
where $[h] = h + \delta C_1$ is the homology class of $h$. 

The standard flux functional associated to GBP is projectively faithful 
\cite[prop 5.17]{phd} at $\beta = 1$. It is defined for every 
pair $\aa \supset \bb$ in $K$ by: 
\begin{equation}
\Phi(v)_{\aa \to \bb}(x_\bb) = 
- \ln \sum_{x_{\aa |\bb} = x_\bb} e^{V_\bb(x_\bb) -V_\aa(x_\aa)}
\quad {\rm where}\quad 
V_\aa(x_\aa) \sum_{\bb \incl \aa} v_\bb(x_{\aa | \bb})  
\end{equation}
One can nonetheless build better functionals 
by making use of Bethe-Kikuchi coefficients \cite{gsi21}, faithful 
at least in a neighbourhood of $\Fix$ \cite[prop 5.37]{phd}.
A consequence of faithfulness is that
{\it stationarity of beliefs implies
stationarity of messages}\footnote{
    Faithfulness hence consists of a transversality property 
    between $\Img(\Phi_*)$ and $\Ker(\delta)$. 
}. One can hence shift the focus from messages 
to beliefs (or from flux terms to potentials), without losing information 
on GBP equilibria. 

\begin{Definition} 
For every $\aa \incl \Om$ and $\beta > 0$, 
let us introduce the following local functionals: 
\begin{itemize}
    \item $S_\aa : \Prob(E_\aa) \to \R$ 
    the {\em Shannon entropy} defined by 
    $S_\aa(p_\aa) = - \sum p_\aa \ln(p_\aa)$,
    \item $\varF^\beta_\aa : \Prob(E_\aa) \times \R^{E_\aa} \to \R$
    the {\em variational free energy} defined by 
    $\varF^\beta_\aa(p_\aa, H_\aa) = \E_{p_\aa}[H_\aa] - \frac 1 \beta S_\aa(p_\aa)$,
    \item $\F^\beta_\aa : \R^{E_\aa} \to \R$ 
    the {\em free energy} defined by 
    $\F^\beta_\aa(H_\aa) = - \frac 1 \beta \ln \sum \e^{- \beta H_\aa}$,
\end{itemize}
\end{Definition}

Recall that $\F^\beta_\aa(H_\aa) = \min_{p_\aa} \varF^\beta_\aa(p_\aa, H_\aa)$
is the Legendre transform of Shannon entropy 
\cite{Jaynes-57,Marle-16,Bennequin-IEM}, 
the global optimum satisfying 
$p^*_\aa = \rho^\beta_\aa(H_\aa) = \F^\beta_{\aa*}(H_\aa)$, i.e.
the differential $F_{\aa*}^\beta : \R^{E_\aa} \to \R^{E_\aa *}$  
of free energy yields Gibbs states $p^*_\aa \in \Prob(E_\aa)$.

The Bethe-Kikuchi approximation estimates
the global variational free energy $\varF^\beta_\Om$ by 
a sum of local free energy cumulants $c_\aa \varF^\beta_\aa$. 
This consists of a truncated Möbius inversion \cite{Morita-57},
whose error decays exponentially as
$K \incl \Pow(\Om)$ grows coarse with respect to the range of interactions
\cite{Schlijper-83}.

\begin{Problem} \label{varF-crit}
Let $H = \zeta h \in C_0$ denote local energies 
and chose an inverse temperature $\beta > 0$.
\\
Find $p \in \Gamma_0$ critical for the 
{\em Bethe-Kikuchi variational energy} 
$\check \varF : \Delta_0 \times C_0 \to \R$ given by:
\begin{equation}
\check\varF^\beta(p, H) = \sum_{\aa \in K} c_\aa \: \varF^\beta_\aa(p_\aa, H_\aa) 
\end{equation}
\end{Problem}

The theorem of Yedidia, Freeman and Weiss \cite{Yedidia-2005} 
states that solutions of problem \ref{varF-crit} are found by their GBP 
algorithm. We shall precise this correspondence
\cite[theorem 4.22]{phd} below and 
now introduce the two dual variational principles on entropy 
and free energy we are concerned with.

As in the exact global case, the max-entropy principle takes place with constraints 
on the Bethe-Kikuchi mean energy:
\begin{equation}
\check U(p, H) = \sum_{\aa \in K} c_\aa \: \E_{p_\aa}[H_\aa]
= \sum_{\aa \in K} \E_{p_\aa}[h_\aa] = \langle p, h \rangle
\end{equation}
For consistent $p \in \Gamma_0$, 
the mean energy thus computed is exact by \ref{U_Bethe} and linearity
of integration,
errors in Bethe-Kikuchi approximations only coming
from the non-linearities of information functionals.

\begin{Problem} \label{S-crit}
Let $H = \zeta h \in C_0$ denote local energies and 
chose a mean energy $\cal U \in \R$.
\\
Find $p \in \Gamma_0$ constrained to $\langle p, h \rangle = \U$ 
and critical for the {\em Bethe-Kikuchi entropy}
$\check S : \Delta_0 \to \R$ given by:
\begin{equation}
    \check S(p) = \sum_{\aa \in K} c_\aa \: S_\aa(p_\aa)
\end{equation}
\end{Problem}

Problems \ref{varF-crit} and \ref{S-crit} are 
both variational principles on $p \in \Gamma_0$,
with the common consistency constraint $dp = 0$ 
but dual temperature and energy constraints respectively.
The following free energy principle instead explores
interaction potentials $v \in C_0$ satisfying a global 
energy conservation constraint $[v] = [h] \in C_0 / \delta C_1$
yet at a fixed temperature $T = \beta^{-1}$.

\begin{Problem} \label{F-crit}
Let $H = \zeta h \in C_0$ denote local energies and chose an inverse
temperature $\beta > 0$. 
\\
Find $V = \zeta v \in C_0$ such that $v \in h + \delta C_1$,
critical for the {\em Bethe-Kikuchi energy} $\check\F : C_0 \to \R$
given by: 
\begin{equation} \label{F_Bethe}
    \check \F^\beta(V) = \sum_{\aa \in K}
    c_\aa \: F^\beta_\aa(V_\aa)
\end{equation}
\end{Problem}

We will now show that problems \ref{varF-crit}, \ref{S-crit} and \ref{F-crit}
are equivalent to solving local consistency constraints on the beliefs 
$p = \rho^\beta(\zeta v)$ 
induced by potentials $v \in h + \delta C_1$ that satisfy the energy conservation 
constraints.
\section{Correspondence theorems}

Solutions of the local max-entropy and free energy principles 
share a very common structure.
Dual mean energy and
temperature constraints may however not lead to a univocal 
relationship between $\U$ and $\beta$, 
as is the case for their global counterparts. 

\begin{Theorem} \label{thm1}
    
    Let $\beta > 0$ and $h \in C_0$.
    Under the correspondence $p = \rho^\beta(\zeta v)$,
    problem \ref{varF-crit} is equivalent to finding 
    consistent conservative potentials $v \in [h] \cap \Fix^\beta$.

\end{Theorem}

\begin{Theorem} \label{thm2}
    Let $\U \in \R$ and $h \in C_0$.
    Under the correspondence $p = \rho^1(\zeta \bar v)$, 
    problem \ref{S-crit} is equivalent to finding
    consistent conservative potentials
    $\bar v \in [\beta h] \cap \Fix^1$ for some $\beta > 0$
    and such that $\langle p, h \rangle = \U$.
\end{Theorem}  

Let $c : C_0 \to C_0$ denote multiplication by Bethe-Kikuchi coefficients.
This operation may not be invertible in general, as $c_\bb$ may vanish
on some $\bb \in K$. In fact, one knows that $c_\bb = 0$ whenever $\bb$ 
is not an intersection of maximal regions $\aa_1, \dots, \aa_n \in K$, 
although assuming that $K$ is the $\cap$-closure of 
a set of maximal regions does not always imply 
invertibility of $c$. 

\begin{Definition}
Let us call 
$\Fixt^\beta = \{v + b \:|\: 
v \in \Fix^\beta, b \in \Ker(c \zeta)\} \incl C_0$ 
the {\em weakly consistent manifold}. 
In particular, $\Fixt^\beta = \Fix^\beta$ when $c$ is invertible. 
\end{Definition}

\begin{Theorem} \label{thm3}
    Given $\beta > 0$ and $h \in C_0$, 
    problem \ref{F-crit} is equivalent to finding 
    weakly consistent conservative potentials 
    $w \in [h] \cap \Fixt^\beta$
    and they can be univocally mapped onto $[h] \cap \Fix^\beta$
    by a retraction $r^\beta : \Fixt^\beta \to \Fix^\beta$. 
\end{Theorem}

The retraction $r^\beta$ will be defined
by equation \ref{W} in the proof below, it maps solutions of 
problem \ref{F-crit} 
onto those of \ref{varF-crit}. We now proceed to prove theorems 
\ref{thm2} and \ref{thm3}.
We shall discuss their relationship afterwards. The following 
combinatorial lemma is essential to the proofs and rather subtle
in spite of its apparent simplicity (see \cite[chap 4]{phd} for
detailed formulas). 

\begin{Lemma} \label{lemma1}
    There exists a linear flux map $\Psi : C_0 \to C_1$ such that
    $c - \mu = \delta \Psi$.
\end{Lemma}

\begin{proof}[Proof of lemma \ref{lemma1}]
\nocite{Kellerer-64,Matus-88} 
One can completely characterise
homology classes $[h] = h + \delta C_1$ by the global energy conservation constraint 
$\sum_\aa h_\aa = H_\Om$ when $K \incl \Pow(\Om)$ is closed under 
intersection \cite[cor 2.14]{phd}.
Therefore equation \ref{U_Bethe} implies that $h = \mu H$ and $c H$ 
are homologous, i.e. $(c - \mu)(H) \in \delta C_1$ and 
one can construct $\Psi : C_0 \to C_1$ 
such that $(c - \mu)(H) = \delta \Psi(H)$.  
\end{proof}

\begin{proof}[Proof of theorem \ref{thm1}]
Very similar to that of \ref{thm2} below, the precise proof can 
be found in \cite[thm 4.22]{phd}.
\end{proof}

\begin{proof}[Proof of theorem \ref{thm2}]
Given $\U \in \R$ and $h \in C_0$, let $p \in \Gamma_0$ 
such that $\langle p, h \rangle = \U$ be a solution of 
\ref{S-crit}. Define local energies $\bar V = \zeta \bar v \in C_0$ 
by letting $\bar V_\bb = - \ln p_\bb$ for all $\bb$ so that 
$p = \rho^1(\bar V) = \rho^1(\zeta \bar v)$.  
The consistency of $p$ therefore implies $\bar v \in \Fix^1$. 
This holds for any other potential $\bar v' \in \bar v + \R^K$, 
as addition of local energy constants preserves Gibbs states. 

Let us now show that the energy constraint $\langle p, h \rangle = \U$ 
and the consistency constraints $dp = 0$ imply 
$\bar v \in [\beta h]$ for some $\beta$. 
By duality, recalling that $\delta = d^*$ we know that 
$\langle q, u \rangle = 0$ for all $q \in \Ker(d)$ 
is equivalent to $u \in \Img(\delta)$.
The consistency constraint hence introduces Lagrange multipliers 
$\delta \ph \in \delta C_1$. 
The linear constraints $\langle p_\aa, 1_\aa \rangle = 1$ and $\langle p, h \rangle = \U$
respectively introduce Lagrange multipliers terms of the form 
$\lambda + \beta h \in \R^K  + \R h$.

A classical computation yields the differential 
of local entropies $S_{\bb *}(p_\bb) = - \langle \cdot, 1_\bb + \ln(p_\bb)\rangle$, 
which coincides with $\langle \cdot , \bar V_\bb\rangle$ 
on tangent fibers of $\Prob(E_\bb)$ 
due to the normalisation constraint. 
Therefore $p$ is critical if and only if:
\begin{equation} \label{dS}
\check S_*(p)_\aa = c_\aa (\bar V_\aa - 1_\aa) = \lambda_\aa 1_\aa + \beta h_\aa + \delta \ph_\aa 
\end{equation}
By lemma \ref{lemma1} we know that $c \bar V$ and $\bar v = \mu \bar V$
are homologous, 
i.e. $\bar v = c \bar V - \delta \Psi(c \bar V)$. Therefore 
\ref{dS} is equivalent to $\bar v \in \lambda + \beta h + \delta C_1$.
Up to a boundary term of $\delta C_1$, the local constants $\lambda \in \R^K$ 
could also be absorbed into a single constant $\lambda' \in \R$. 
Enforcing the constraint $\langle p, \bar v' \rangle = \beta \U$ on 
the equivalent potentials $\bar v' \in \bar v + \R^K$, 
one may furthermore ensure that $\lambda' = 0$ 
so that $\bar v' \in [\beta h] \cap \Fix^1$.
\end{proof}

\begin{proof}[Proof of theorem \ref{thm3}]
Given $H = \zeta h \in C_0$ and $\beta > 0$, 
let $V = \zeta v \in C_0$ be a solution of \ref{F-crit}. 
The energy constraint on $v \in h + \delta C_1$ 
implies $\check F^\beta_*(V)_{|\Img(\zeta\delta)} = 0$ with
the differential $\check F^\beta_* : C_0 \to C_0^*$ given by: 
\begin{equation}
    \check \F^\beta_*(V)_\aa = c_\aa \: \rho^\beta(V_\aa)
\end{equation}
For any subspace $B \incl C_0$, let us write $B^\perp \incl C_0^*$ for
the orthogonal dual (or annihilator) of $B$.
Letting $p = \rho^\beta(V)$
criticality is then equivalent to $\F^\beta_*(V) = c p \in \Img(\zeta \delta)^\perp$.
Recall that $d = \delta^*$ by definition and denote by 
$\zeta^* : q \mapsto q \circ \zeta$ the adjoint of $\zeta$: 

\begin{equation}
    cp \in \Img(\zeta\delta)^\perp 
    \quad\eqvl\quad   \zeta^*(cp) \in \Img(\delta)^\perp
    \quad\eqvl\quad   \zeta^*(cp) \in \Ker(d)
\end{equation}  
Therefore $V$ is critical for $\check \F^\beta_{|\Img(\zeta \delta)}$ 
if and only if $\zeta^*(c p) = \zeta^*(c \,\rho^\beta(V))$ is consistent.

Assume $p \in \Ker(d)$ is consistent. Then by the dual form of 
lemma \ref{lemma1} involving $(\delta \Psi)^* = \Psi^* d$ 
which vanishes on $\Ker(d)$, we have
$c p = \mu^* p$ 
and $\zeta^*(cp) = p$ is consistent as well. 
This shows that any consistent potential $v \in [h] \cap \Fix^\beta$ 
yields a critical point $V = \zeta v$ of the Bethe-Kikuchi energy 
$\check \F^\beta(V)$.

Reciprocally assume $q = \zeta^*(c p) \in \Ker(d)$, 
then $\mu^*(q) = c q$ by lemma \ref{lemma1} and $c q = c p$,
which means that $p_\bb$ and $q_\bb$ coincide on any $\bb \in K$ 
such that $c_\bb \neq 0$. 
Let us define energies $W = \zeta w \in C_0$ by: 
\begin{equation} \label{W}
W_\bb = - \frac 1 \beta \ln(q_\bb) + \F_\bb^\beta(V_\bb)
\end{equation}
Observing that $q = \rho^\beta(W)$ and
$\F^\beta_\bb(W_\bb) = \F^\beta_\bb(V_\bb)$ for all $\bb$,
one sees that $c p = c q$ implies $c V = c W$. 
By lemma \ref{lemma1} we get
$\mu V \in \mu W + \delta C_1$ so that the potentials $w = \mu W$ 
satisfy the global energy constraint $[w] = [v] = [h]$.
By the assumption $q \in \Ker(d)$ 
this shows that $w \in [h] \cap \Fix^\beta$. 
The relation $cV = cW$ implies
$v - w = \mu (V - W) \in \Ker(c\zeta)$, 
so that $v \in [h] \cap \Fixt^\beta$ 
is weakly consistent by definition of $\Fixt^\beta$. 

The retraction $r^\beta : \Fix_+^\beta \to \Fix^\beta$ 
is described more succinctly via its conjugate 
$R^\beta = \zeta \circ r^\beta \circ \mu$ acting
on local energies $V = \zeta v \in C_0$:
\begin{equation}
    R^\beta(V)_\bb = - \frac 1 \beta
    \ln \sum_{\aa \cont \bb} c_\aa \: \pi_*^{\aa \to \bb}
    \big(\rho^\beta_\aa(V_\aa) \big) 
    - \F_b^\beta(V_\bb)
\end{equation}
which amounts to applying the linear projection $\zeta^* c$ on beliefs
before chosing energies as \ref{W}.
\end{proof}

\section{Conclusion}

\vspace{0.3cm}
\begin{figure}[t]
    \sbox0{\includegraphics[width=\textwidth]{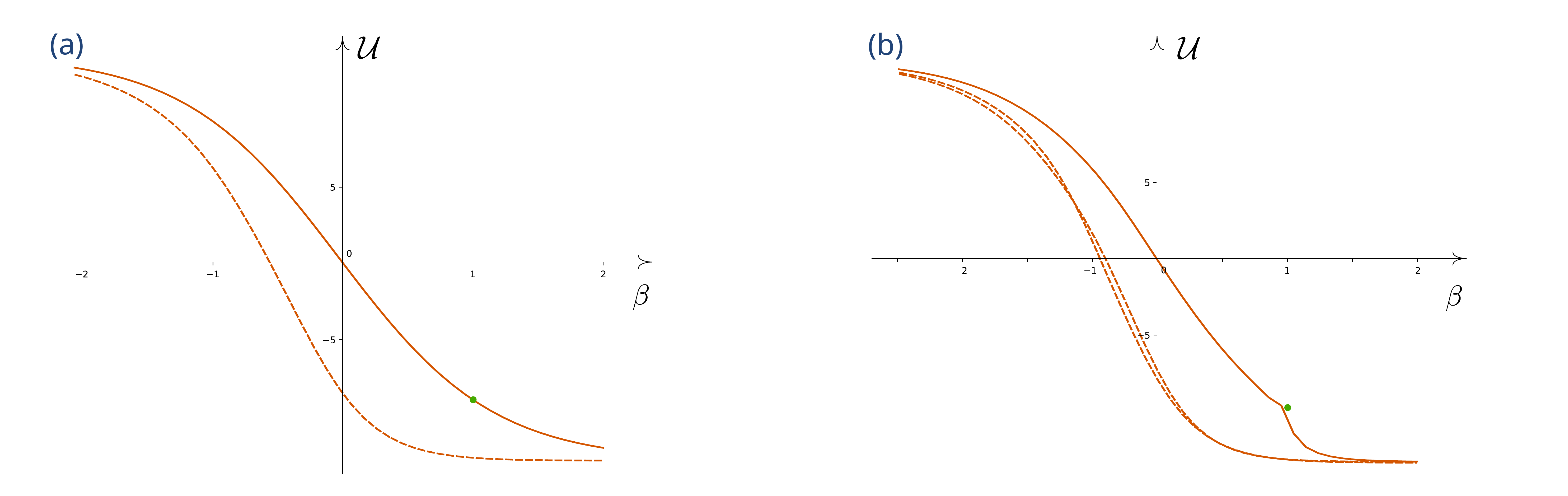}}
\begin{center}
    \begin{minipage}{\textwidth}
        \vspace{0.2cm}

\centering
\vspace{1cm}
\usebox0
    \vspace{-0.2cm}
    \caption{
        Energy-temperature relationships on a graph with two loops. The initial potential $\beta h$ (solid line) 
        reaches (a) a cuspidal singularity and (b) a generic singularity at $\beta = 1$.
        Dashed lines correspond to mean energies of stable equilibria obtained from homologous 
        initial conditions $\beta h + \delta \ph$.
        Note that $\beta < 0$ should be understood as a reversal of interactions, 
        the limit $\beta \to \pm0$ being equivalent to $T \to \infty$.
    }
\end{minipage}
\end{center}
    \vspace{-0.5cm}
\end{figure} 

Letting $\bar v = \beta v$, note that
solutions of \ref{S-crit} can be written 
$p = \rho^\beta(\zeta v)$ with 
$v \in [h] \cap \Fix^\beta$ solution of \ref{varF-crit},
as $\bar v \in [\beta h] \cap \Fix^1$ 
is equivalent to $v \in [h] \cap \Fix^\beta$.
The form of theorem \ref{thm2} involves a simpler intersection problem 
in $C_0$, 
between the linear subspace $(\beta h)_{\beta \in \R} + \delta C_1$,
a non-linear energy constraint 
$\langle \rho^1(\zeta \bar v), h\rangle = \U$,
and the consistent manifold $\Fix^1$ that does not depend on $\beta$. 
The possible multiplicity of solutions for a given energy 
constraint $\U$ may however occur for different values of the Lagrange 
multiplier $\beta$.

In the high temperature limit $\beta \to 0$,
one can show that $[h]$ is transverse to $\Fix^\beta$
\cite[prop 5.11]{phd}, or equivalently 
that $\Img(\delta) \cap T_0 \Fix = \{ 0 \}$, 
so that a univocal energy-temperature relationship $\beta_h(\U)$ 
can be defined in a neighbourhood ${\cal V}_0$ of $h = 0$.
In this quasi-linear regime, all three problems 
are therefore equivalent 
to finding the intersection $[h] \cap \Fix^{\beta_h(\U)}$ restricted to ${\cal V}_0$.

Let us say that $\bar v \in [\beta h] \cap \Fix^1$ is 
{\it singular} if $T_v \Fix^1 \cap \delta C_1 \neq 0$, which may occur 
for some $\beta = \beta_c$.
A generic singular point, although consistent, is not stable 
(cuspidal singularities are one exception). 
The diffusion flow will therefore generically depart from a singular
potential, one should then expect a sharp drop in internal energy 
for $\beta > \beta_c$ as figure 1 illustrates empirically on a graph with two loops.
In general, one may define an energy spectrum 
$\bar \U_h(\beta) : \R_+^* \to \Pow_{\rm fin}(\R)$ 
as the image of $[h] \cap \Fix^\beta$ 
under $v \mapsto \langle \rho^\beta(\zeta v), h \rangle$ to account for the multiplicity 
of equilibria. Figure 1, 
obtained in the Ising model on a graph with two loops, 
illustrates a case where $\bar \U_h(0)$ 
and $[0] \cap \Fix$ are degenerate of cardinal $\geq 3$. 
The singular space can be parameterised by polynomial equations on 
$p = \rho(\zeta h) \in \Gamma_0$ which can be solved exactly in simple examples 
\cite[chap 6]{phd}. They compute 
the determinant of the linearised diffusion flow 
restricted to $\delta C_1$.



\reftitle{References}

\bibliography{biblio/biblio.bib}

\end{document}